\documentclass[11pt]{article}

\usepackage{bbm}

\usepackage[T1]{fontenc}
\usepackage[utf8]{inputenc}
\usepackage{lmodern}
\usepackage{amsmath,amssymb,amsthm}
\usepackage{fullpage}

\newtheorem{theorem}{Theorem}

\newtheorem{lemma}[theorem]{Lemma}

\newtheorem{statement}[theorem]{Statement}
\newtheorem{conjecture}[theorem]{Conjecture}

\usepackage{algorithm}
\usepackage{algpseudocode}

\newcommand{\Prob}[1]{\mathbf{Pr}\left[#1\right]}

\def\epsilon{\ensuremath{\varepsilon }}
\newcommand{\eps}{\ensuremath{\epsilon }}

\newcommand{\Exp}[1]{\mathbf{E}\left[#1\right]}

\newcommand{\COMMENTED}[1]{{}}

\newcommand{\real}{\ensuremath{\mathbb{R}}}
\usepackage{graphicx}

\begin{document}
\title{Collecting Coupons is Faster with Friends}
%
%
\author{Dan Alistarh\\IST Austria\\\texttt{dan.alistarh@ist.ac.at} \and
 Peter Davies\\IST Austria\\\texttt{peter.davies@ist.ac.at}}
%
%
%
\maketitle              

\begin{abstract}
In this note, we introduce a distributed twist on the classic coupon collector problem: 
a set of $m$ collectors wish to each obtain a set of $n$ coupons; for this, they can each sample coupons uniformly at random, but can also  meet in pairwise interactions, during which they can exchange coupons. By doing so, they hope to reduce the number of coupons that must be sampled by each collector in order to obtain a full set. This extension is natural when considering real-world manifestations of the coupon collector phenomenon, and has been remarked upon and studied empirically [Hayes and Hannigan 2006, Ahmad et al. 2014, Delmarcelle 2019].

We provide the first theoretical analysis for such a scenario. We find that ``coupon collecting with friends'' can indeed significantly reduce the number of coupons each collector must sample, and raises interesting connections to the more traditional variants of the problem. While our analysis is in most cases asymptotically tight, there are several open questions raised, regarding finer-grained analysis of both ``coupon collecting with friends,'' and of a long-studied variant of the original problem in which a collector requires multiple full sets of coupons.

\end{abstract}

\section{Introduction}
The coupon collector problem is a classic exercise in probability theory, appearing in standard textbooks such as those of Feller \cite{F57} and Motwani and Raghavan \cite{MR95}. It is often introduced with a story along the lines of the following: a cereal company runs a promotion giving away a toy (the ``coupon'') in each box of cereal sold. The toys are chosen uniformly at random from some finite set of different types. A child wishes to collect the full set of toys, and our task is to analyze the number of cereal boxes her parents must purchase to achieve this. This number is, of course, a random variable, and while elementary bounds on it are quite straightforward, a tighter analysis requires more sophisticated techniques (see, e.g., \cite{ER61}).

A modern real-world example of this phenomenon is the World Cup sticker album \cite{HH06}. Collectors purchase sealed packs of stickers of football players, and aim to collect one of each in order to fill all the slots in their album. Completing the sticker album has proven a very popular activity among (mostly, but by no means exclusively) young football fans every four years, and has highlighted an aspect which is absent from the classical analysis of the coupon collector's problem: one can achieve a full collection much faster by swapping duplicate coupons with friends who are also collecting. This has been noted previously and studied empirically, specifically for the World Cup sticker album \cite{HH06,ahmadcoupon,D19}, but we are not aware of any prior theoretical analysis for such a setting in general.

Of course, for theoretical analysis, one must first define a model specifying how swapping of coupons is permitted. If all collectors are allowed to swap freely, then the problem is equivalent to a variant which \emph{has} seen prior theoretical study: that in which $m>1$ full sets of coupons must be completed by a single collector. This variant was studied by Newman and Shepp \cite{N60} and Erdös and Rényi \cite{ER61}. In more recent works (e.g., \cite{DP18} and the references therein), problem settings of this sort are often referred to as ``coupon collector with \emph{siblings}:'' the accompanying story is that there is a single collector, but she has a succession of younger siblings to whom she gives duplicate coupons upon receiving them. One can then ask long it takes for the $m^{th}$ sibling to complete his collection. 
Specifically, Newman and Shepp \cite{N60} showed that the number of coupons needed to complete $m$ full sets is $n(\log n+(m- 1) \log\log n+O(1))$ in expectation. Erdös and Rényi \cite{ER61} provided concentration bounds around this expectation, and specified the constant in the linear term. However, it is important to note that this bound holds \emph{only when $m$ is a constant}; as Erdös and Rényi themselves note, ``It is an interesting problem to investigate the limiting distribution of $v_m(n)$ when $m$ increases together with $n$, but we can not go into this question here.'' Surprisingly, to our knowledge, this open problem has never been addressed, and while we give an asymptotic analysis here, it remains an open question to extend the more fine-grained bounds of Newman and Shepp, and Erdös and Rényi to the case where $m$ also tends to infinity.

Our primary focus is a \emph{distributed} generalization: when completing, for example, the World Cup sticker album, collectors generally do not, to the authors' knowledge,  deliberately congregate in large groups in order to exchange stickers in an organized fashion. Instead, we would expect that exchanges are usually ad-hoc, and made between individual pairs of collectors. So, we will abstract such behavior using a ``population protocol''-style model of random pairwise interactions: in each round, an independent, uniformly random pair of collectors will meet, and can swap coupons between them as they wish. We then aim to analyze the trade-off between the number of coupons that each collector must sample, and the number of interactions required, in order for all collectors to obtain full collections. We call this problem ``coupon collecting with friends.''

\subsection{The Formal Problem Setting}
A set $M$ of $m$ collectors each wish to obtain a full collection of $n$ distinct coupons. For this, they will operate in sequences of \emph{collection} (sampling) and \emph{exchanging} (interaction) phases:

\begin{enumerate}
	\item A collection phase, in which each collector independently and uniformly samples, with replacement, $r_c$ coupons from $[n]$.
	\item An exchanging phase, in which $r_e$ sequential interactions between independent, uniformly random pairs of collectors occur. An interacting pair of collectors can choose to exchange coupons however they wish.
\end{enumerate}

We are interested in the trade-off between the numbers of collection rounds and exchanging rounds ($r_c$ and $r_e$) that are required for each of the $m$ agents to obtain a full collection of $n$ distinct coupons.

\subsection{Preliminaries}
In the following, we denote $\ln x := \log_e x$ and $\log x := \log_2 x$. We make frequent use of the well-known inequalities $1-x\le e^{-x}$ for $x\in\real$ and $1-x\ge 4^{-x}$ for $x \in [0, \frac 12]$, and the Chernoff bound in the following standard form:

\begin{lemma}[Chernoff bound]\label{lem:Chernoff}
	Suppose $Z_1, \dots, Z_t$ are independent random variables taking values in $\{0, 1\}$. Let $Z$ denote their sum and let $\mu = \Exp{Z}$ denote the sum's expected value. Then for any $\delta \in [0,1]$,
	
\[\Prob{Z \le (1-\delta)\mu} \le e^{-\frac{\delta^2\mu}{2}}\enspace.\]
\end{lemma}

\section{What Happens With No Exchanges?} 
We first look at the most ``standard'' variant of the trade-off: when $r_e = 0$, i.e., no exchanges are allowed. In this case, the problem is simply $m$ separate instances of the standard coupon collector problem, since each collector must independently collect a full set without help from the other collectors. 

It has long been known \cite{N60,ER61} that the number of samples needed for a single collector to obtain a full set is $n \ln n \pm O(n)$ with probability $1-\eps$, where $\eps>0$ is any positive constant. To be precise, we use the following statement as phrased by Motwani and Raghavan: 

\begin{statement}[\cite{MR95}, corollary to Theorem 3.8, Section 3.6.3]\label{st:lb}
	For any real constant $c$, we have
	\[\lim\limits_{n\rightarrow \infty} \Prob{X \le n(\ln n - c)} = e^{-e^c}\]
	\center{and}
	\[\lim\limits_{n\rightarrow \infty} \Prob{X \ge n(\ln n + c)} = 1-e^{-e^{-c}}\enspace.\]
\end{statement} 

(Here $X$ is the random variable denoting the number of required samples.) The statement implies that the probability of failure for a single collector after $n\ln n + \omega(n)$ samples tends to $0$ as $n$ tends to infinity. However, this is not quite sufficient for us: we require $m$ independent instances to all succeed, so we need the probability of failure for each collector to be less than $1/m$, and we do not treat $m$ as a constant. So, we need to know how \emph{fast} the failure probability tends to $0$.

We give the following straightforward asymptotic upper and lower bounds for the problem (for $n>1$; for $n=1$, exactly $1$ collection round is clearly necessary and sufficient).
\begin{lemma}
If $r_e = 0$, then $r_c= O(n \log mn)$ is sufficient to succeed with probability $1-(mn)^{-1}$.
\end{lemma}

\begin{proof}
Let $r_c = 2n \ln mn$. Fix a particular collector $v$ and coupon $\alpha$. The probability that $v$ does not collect a copy of $\alpha$ is at most $\left(1-\frac 1n\right)^{2n\ln mn}\le e^{-\frac{2n\ln mn}{n}} = (mn)^{-2}$. By a union bound over all coupons and collectors, the probability that any collector does not receive a copy of any coupon is at most $mn \cdot (mn)^{-2} = (mn)^{-1}$.
\end{proof}

\begin{lemma}
	For $n>1$, if $r_e = 0$, then $r_c= \Omega(n \log mn)$ is necessary to succeed with any positive constant probability.
\end{lemma}

\begin{proof}
By Statement \ref{st:lb}, even a single collector must perform $\Omega(n\log n)$ samples to collect all $n$ coupons with any constant probability. We now show that $\Omega(n\log m)$ samples per collector are required for all $m$ collectors to be successful. The lower bound is then $\Omega(\max\{n\log m,n\log n\})= \Omega(n\log mn)$.

Fix a particular coupon $\alpha$, and let $r_c \le \frac{1}{4} n \log m$. The probability that a particular collector $v$ does not receive a copy of $\alpha$ is $(1-\frac 1n)^{r_c} \ge 4^{\frac{-r_c}{n}} \ge 4^{\frac{-\log m}{4}} = m^{-\frac{1}{2}}$ (using that $1-x\ge 4^{-x}$ for $x \in [0,\frac 12]$).

The events that each collector receives a copy of $\alpha$ are independent. Therefore, the probability that all collectors receive a copy is 

\begin{align*}
\Prob{\bigcap_{v \in M}\text{\{$v$ receives a copy of $\alpha$\}}} &
= \prod_{v \in M}\Prob{\text{$v$ receives a copy of $\alpha$}}\\
&\le \prod_{v \in M} \left(1-m^{-\frac{1}{2}}\right)\\
&= \left(1-m^{-\frac{1}{2}}\right)^m\\
&= e^{-m^{-\frac{1}{2}} \cdot m} = e^{-\sqrt m}.
\end{align*}

So, in order to achieve any constant (as $m\rightarrow \infty$, which we may assume since this component of the lower bound is only relevant when $m>n$) probability of success, we require $r_c> \frac{1}{4} n \log m$. 
\end{proof}

\section{What Happens With Unlimited Exchanges?}
If an unlimited amount of exchanges are allowed, then the problem is equivalent to simply ensuring that $m$ copies of each coupon are sampled between all collectors, since the exchanges will then allow these to eventually be distributed to each collector. As mentioned, for constant $m$ strong bounds are known \cite{N60,ER61}, but we are not aware of any prior work for non-constant $m$.

Again, we can show straighforward matching asymptotic bounds:

\begin{lemma}
	If $r_e = \infty$, then $r_c= O(n+\frac{n\log n}{m})$ is sufficient to succeed with probability $1-\frac 1n$.
\end{lemma}

\begin{proof}
Let $r_c= 16(n+\frac{n\ln n}{m})$, i.e., $16(mn+n\ln n)$ total samples are taken. Fix a particular coupon $\alpha$. The expected number of copies of $\alpha$ obtained is $\mu:=16(m+\ln n)$, and each sample is independent. So, by a Chernoff bound (Lemma \ref{lem:Chernoff}),
	
	\begin{align*}
		&\Prob{\text{fewer than $m$ copies of $\alpha$ are collected}}\\
		&\hspace{1in}< \Prob{\text{at most $(1-\frac 12)\mu$ copies of $\alpha$ are collected}}\\
		&\hspace{1in}\le e^{-\frac{\mu}{8}} \le e^{-\frac{16\ln n}{8}}\le n^{-2}.
	\end{align*}
	
	Taking a union bound over all coupons, we find that the probability that any coupon does not have at least $m$ copies sampled is at most $\frac 1n$. So, with probability at least $1-\frac 1n$, every coupon is sampled at least $m$ times, and with unlimited exchanges we can complete every collector's collection.
\end{proof}

\begin{lemma}
	If $r_e = \infty$, then $r_c= \Omega(n+\frac{n\log n}{m})$ is necessary to succeed with any positive constant probability.
\end{lemma}

\begin{proof}

A lower bound of $\Omega(n\log n)$ samples follows from the standard coupon collector problem: by Statement \ref{st:lb}, with $o(n\log n)$ samples we cannot collect even one copy of all coupons with any constant probability. Furthermore, $mn$ samples are clearly necessary to collect $m$ copies of each of the $n$ coupons. So, we have a lower bound of $\Omega(\max\{mn,n\log n\})$ total samples, i.e. $r_c=  \Omega(n+\frac{n\log n}{m})$.
\end{proof}

We now see the power of allowing exchanges: with unlimited exchanges between $m$ participants, the amount of samples required per collector reduces from $\Theta(n\log mn)$ to $\Theta(n+\frac{n\log n}{m})$. In particular, collaborating with a small group of $m=O(\log n)$ collectors reduces the required number of samples linearly in $m$ (from $\Theta(n\log n)$ to $\Theta(\frac{n\log n}{m})$), which may be an appealing prospect to collectors of World Cup stickers (or their parents).

\section{Minimizing Exchanges for Optimal Collection Rounds}
Now we reach the main question of this work: how many exchanging rounds are necessary to ensure completion using the asymptotically optimal amount of collection rounds? 

We first prove the following upper bound:

\begin{theorem}\label{thm:mainub}
If $r_c\ge 36(n+\frac{n\ln n}{m}) $, $r_e = O(m\log mn)$ is sufficient to succeed with probability $1-\frac {1}{mn}$.
\end{theorem}
\begin{proof}
	Fix a coupon $\alpha$ to analyze during the collection phase. We will call collectors that receive fewer than $2$ copies of $\alpha$ during the collection phase \emph{bad}, and call them \emph{good} otherwise. Fix also a specific collector $v$. After $r_c$ collection rounds, $v$ receives, in expectation, $\mu:=\frac{r_c}{n} $ copies of $\alpha$. Every collection round is independent, so by a Chernoff bound, the probability that $v$ is bad is at most:
	
	\begin{align*}
		\Prob{\text{$v$ is bad}}&
		=\Prob{\text{$v$ receives at most $\left(1-\frac{\mu-1}{\mu}\right)\mu$ copies of $\alpha$}}\\
		&\le e^{-\frac{(\mu-1)^2 }{2\mu}} < e^{1-\frac \mu 2}.
	\end{align*}
	
	 Keeping $\alpha$ fixed but unfixing $v$, we now wish to bound the probability that at least $\frac m3$ collectors are bad (have fewer than $2$ copies of $\alpha$). We do this by a union bound over all possible sets of $\frac m3$ collectors (technically $\lceil\frac m3\rceil$, but we omit the ceiling functions for clarity since the effect is negligible), using that the `badness' of collectors is independent:
	 
	\begin{align*}
	\Prob{\text{at least $\frac m3$ collectors are bad}}&
	=\Prob{\bigcup_{\substack{S\subset M\\|S|=\frac m3}}\text{all collectors in $S$ are bad}}\\
	&\le \sum_{\substack{S\subset M\\|S|=\frac m3}}\Prob{\text{all collectors in $S$ are bad}}\\
	&\le \binom{m}{\frac m3} (e^{1-\frac \mu 2})^{\frac m3}\\
	&\le \left(3e\right)^\frac m3 e^{(1-\frac \mu 2)\frac m3}\\
	&=e^{-(\frac \mu 2-2-\ln 3)\frac m3}.
\end{align*}	
	
In the penultimate line here we used the inequality $\binom{a}{b} \le \left(\frac{ae}{b}\right)^b$. Since $\mu = \frac{r_c}{n} \ge 36$, we have $\frac \mu 2-2-\ln 3 > \frac \mu 3$. So, 
\[\Prob{\text{at least $\frac m3$ collectors are bad}}< e^{-\frac \mu 3\cdot \frac m3}= e^{-\frac{r_c m}{9n}}\le e^{-4(m+\ln n)}.\]
	
	Taking a union bound over all coupons, we have that for all coupons there are fewer than $\frac{m}{3}$ bad collectors with probability at least $1-ne^{-4(m+\ln n)} = 1-e^{-4m-3\ln n}$. We call this event a \emph{successful collection phase}.
	
	We now describe the exchanging phase. Let $r_e = 6m\ln mn$. We will use the following simple swapping rule: whenever a collector with at least two copies of some coupon $\alpha$ interacts with a collector with $0$ copies of that coupon, it will give one of its copies. 
	
	A crucial observation is that for a coupon $\alpha$ which has fewer than $\frac m3$ bad collectors after the collection phase, there will always be at least $\frac m3$ collectors with at least two copies throughout the exchanging phase. This is because every time a \emph{good} collector gives away a copy (possibly dropping to $1$ copy itself), a collector with $0$ copies goes up to $1$ copy. Since there are at most $\frac m3$ collectors with $0$ copies to begin with (and collectors never drop down to $0$ copies), this can occur at most $\frac m3$ times, leaving at least $\frac m3$ collectors with multiple copies.
	
	To analyze the exchanging phase, we fix a particular interaction, a particular coupon $\alpha$ which is not yet held by all collectors, and a particular collector $v$ which currently has $0$ copies of $\alpha$. By the above observation, with probability at least $2\cdot \frac 1m\cdot \frac 13 = \frac{2}{3m}$, the interaction pairs $v$ with a collector who has at least $2$ copies of $\alpha$, and so $v$ receives a copy.
	
	There are initially (trivially) at most $mn$ such pairs $(\alpha,v)$ where $v$ holds $0$ copies of $\alpha$. Conditioning on a successful collection phase, each of these pairs is removed in each iteration with probability at least $\frac{2}{3m}$. For a fixed pair, these probabilities hold independently over all iterations. So the probability of a particular pair $(\alpha,v)$ remaining over the entire exchanging phase (i.e., for $v$ to still hold no copy of $\alpha$ upon completion) is at most 
	\[\left(1-\frac{2}{3m}\right)^{6m\ln mn}\le e^{-3\ln mn} = (mn)^{-3}.\]
	
	Taking a union bound over all such pairs, the probability that any pair remains is at most $(mn)^{-2}$. Finally, taking another union bound to remove the conditioning on a successful collection phase, the probability of successfully completing all collections is at least $1- e^{-4m-3\ln n}-(mn)^{-2} \ge 1-\frac {1}{mn}$.

\end{proof}

We next give a pair of lower bounds, which when combined will match the asymptotic expression for $r_e$ from Theorem \ref{thm:mainub}.

\begin{lemma}\label{lem:mainlb1}
	If $r_c \le \frac 14 n\ln n$, then $r_e = \Omega(m\log n)$ is necessary to succeed with probability $1-\frac{1}{n}$.
\end{lemma}
\begin{proof}
	Fix a collector $v$. By Statement \ref{st:lb}, since $r_c = n\log n - \omega(n)$, the probability that $v$ receives a full set of coupons during the collection phase is $o(1)$. Denote this probability $q$. To succeed overall with probability $1-\frac 1n$, there must be some case in which $v$ does not receive a full collection during the collection phase, but gains it during the exchanging phase with probability at least $1-\frac 2n$ (over the randomness in the exchanging phase only), since otherwise the total probability of $v$ having a full collection is at most $q + (1-q)(1-\frac 2n) = 1-\frac 2n + \frac {o(1)}{n}< 1-\frac 1n$ (for sufficiently large $n$).
	
	If $r_e\le \frac{1}{8} m\log n$, and for $m\ge 4$, the probability that $v$ is not involved in \emph{any} interactions is at least 
	\[\left(1-\frac{2}{m}\right)^{r_e} \ge 4^{-\frac2m r_e} \ge 4^{-\frac14 \log n} = n^{-\frac 12}\enspace.\]
	
	In this case $v$ cannot obtain a full collection of coupons if it did not have one after the collection phase. So, the probability of success if $v$ did not gain a full collection during the collection phase is at most $1-n^{-\frac 12} < 1-\frac 2n$, which means that the total success probability is less than $1-\frac 1n$.

	For the remaining case $m<4$, we again apply Statement \ref{st:lb}, which implies that the probability of collecting a single full collection with $n\ln n-\omega(n)$ samples is $o(1)$. In total, over the $m<4$ collectors, we are taking at most $\frac 34 n\ln n = n\ln n -\omega(n)$ samples during the collection phase. So, with probability $1-o(1)$, there is some coupon for which no collector has a copy, in which case we cannot hope to be successful even with $r_e=\infty$. Thus, our overall success probability is $o(1)$.
	
\end{proof}

\begin{lemma}\label{lem:mainlb2}
	If $r_c = n\ln n - \omega(n)$, then $r_e = \Omega(m\log m)$ is necessary to succeed with any positive constant probability.
\end{lemma}
\begin{proof}
By Statement \ref{st:lb}, for $r_c = n\log n - \omega(n)$, the probability that a particular collector $v$ receives a full set of coupons during the collection phase is $o(1)$. The expected number of collectors receiving full sets is therefore $o(m)$. By Markov's inequality, the probability that at least $\frac m2$ collectors receive full sets at most $o(1)$. With probability $1-o(1)$, therefore, there are at least $\frac m2$ collectors without full sets. We call this event an \emph{unsuccessful collection phase}.

To fill each collector's collection overall with any positive constant probability $\eps>0$, there must be at least one instance with an unsuccessful collection phase on which we do so with probability at least $\frac{\eps}{2}$ (over the randomness of the exchanging phase), since otherwise the total success probability would be at most $\frac{\eps}{2} +o(1)$. We will now show that this requires $r_e = \Omega(m\log m)$ exchanging rounds.
 
Assume that we have an unsuccessful collection phase, and a set $S$ of $\frac m2$ collectors without full sets (again omitting ceiling functions for clarity). Fixing some $v\in S$, the probability that each interaction involves $v$ is $\frac{2}{m}$. Furthermore, it is at most $\frac{4}{m}$ independently of the behavior of all other $u\in S$ (the worst case is that all other $u\in S$ are not involved in the interaction, in which case $v$ is involved with probability $\frac{2}{\frac{m}{2}+1} < \frac{4}{m}$). 

If $r_e\le \frac{1}{16} m\log m$, and for $m\ge 8$, the probability that $v$ is not involved in \emph{any} interactions is at least 
\[\left(1-\frac{4}{m}\right)^{r_e} \ge 4^{-\frac4m r_e} \ge 4^{-\frac14 \log m} = m^{-\frac 12}\enspace,\]

independently of the other $u\in S$. Then, the probability that all collectors in $S$ \emph{are} involved in at least one interaction is at most

\[\left(1-m^{-\frac 12}\right)^{|S|}\le e^{-m^{-\frac 12}\frac m2} = e^{-\frac{\sqrt{m}}{2}} = o(1).\]

That is, with probability $1-o(1)$, at least one collector $v$ in $S$ is not involved in any interactions. In this case $v$ cannot obtain a full collection of coupons: by definition of $S$ its collection is incomplete after the collection phase, and it has no interactions in which to gain new coupons in the exchanging phase. So, we have a total success probability of $o(1)$.

The above analysis assumes that $m\rightarrow \infty$; the case $m=O(1)$ is trivial, since by Statement \ref{st:lb}, with probability $1-o(1)$ we have not completed all collections during the collection phase, and so require at least $1 = \Omega(m\log m)$ exchanging rounds.

\end{proof}

Combining Lemmas \ref{lem:mainlb1} and \ref{lem:mainlb2} yields the following theorem:

\begin{theorem}\label{thm:mainlb}
	If $r_c \le \frac 14 n\ln n$, then $r_e = \Omega(m\log mn)$ is necessary to succeed with probability $1-\frac{1}{n}$.
\end{theorem}

\begin{proof}
By Lemmas \ref{lem:mainlb1} and \ref{lem:mainlb2}, we require $r_e = \Omega(\max\{m\log m,m\log n\}) = \Omega(m\log mn)$.
\end{proof}

We make some observations about the bounds we have shown in Theorems \ref{thm:mainub} and \ref{thm:mainlb}. We now know that $\Theta(m\log mn)$ interactions suffice to achieve an asymptotically optimal number of collection rounds, and are necessary to asymptotically improve over the number of samples needed for the standard single-collector case. If one requires a high probability of success in $n$ (i.e. probability at most $\frac 1n$ of failure), these bounds are tight. However, they leave open the possibility of using fewer interactions to achieve a lower (but still at least a positive constant) success probability. In this regime, Lemma \ref{lem:mainlb1} does not apply, so we have only that $O(m\log mn)$ interactions suffice by Theorem \ref{thm:mainub}, and that $\Omega(m\log m)$ are necessary by Lemma \ref{lem:mainlb2}. We conjecture that it is the upper bound that is tight, and the lower bound that could be improved:

\begin{conjecture}
	If $r_c = O(n+\frac{n\ln n}{m})$, then $r_e = \Omega(m\log mn)$ is necessary to succeed with any positive constant probability.
\end{conjecture}	

The reason for this conjecture is that the current lower bound does not take into account the difficulty for collectors with incomplete collections to obtain \emph{multiple} coupons during the exchanging phase; it uses only the hardness of ensuring a single interaction. Since most collectors will have $\Theta(n)$ coupons missing after the collection phase, we would expect that collectors will require some number of interactions depending on $n$ in order to complete their collections. However, since the events of a collector gaining two different coupons from an interaction are not independent, this would require more sophisticated techniques to analyze.

\section{Conclusions and Open Problems}
Our aim in this paper has been to introduce the study of what we argue is a natural distributed variant of the coupon collector problem: collection by a group of collectors which can meet, in random pairwise fashion, to exchange coupons. As mentioned, there is one gap in the asymptotic analysis we provide: whether $o(m\log mn)$ exchanges can suffice for the asymptotic optimum of $\Theta(n+\frac{n\ln n}{m})$ collection rounds, under a weaker success guarantee (than high probability in $n$).

Generally, most of the prior work on the standard coupon collector problem has been on finer-grained analysis, pinning down the exact terms in the number of samples required, and one could ask whether we can do the same here. Such a focus would change the problem significantly: in particular, the approach of Theorem \ref{thm:mainub} (ensuring that a constant fraction of collectors always hold multiple copies of each coupon) would not work if the number of samples was ``only just'' sufficient, and one would need to find a different way to analyze the exchanging phase.

Surprisingly, the situation is still not fully understood, even for the more ``traditional'' case, corresponding to $r_c = \infty$, when $m$ tends to infinity alongside $n$. We therefore close by reiterating the open question posed by Erdös and Rényi, and ask how the coupon collector problem behaves when a non-constant number of full collections are required.

\subsection*{Acknowledgements}
Peter Davies is supported by the European Union’s Horizon 2020 research and innovation programme under the Marie Skłodowska-Curie grant agreement No. 754411. 
\bibliographystyle{plain}
\bibliography{coupon}

\end{document}